\documentclass{article}
\usepackage[utf8]{inputenc}
\usepackage{amsfonts,amsmath,amssymb,amsthm}
\usepackage{authblk} 
\usepackage{verbatim,float,url,enumerate}
\usepackage[format=plain,
            font=it]{caption} 
\usepackage{graphicx,subcaption,epsfig,psfrag}
\usepackage{dsfont,bm,color,appendix, booktabs}
\usepackage{centernot}
\usepackage{natbib}

\setcitestyle{authoryear, open={(},close={)}}
\usepackage[final]{pdfpages}
\usepackage{fancyhdr}
\usepackage{tcolorbox}
\usepackage[margin=1in, lmargin = 1in, rmargin = 1in]{geometry}
\usepackage{pdflscape} 
\usepackage{hyperref}
\usepackage{multirow,multicol}
\usepackage{tabularx}
\usepackage{arydshln}
\usepackage{enumitem}   

\newtheorem{proposition}{Proposition}

\usepackage{algorithm}
\usepackage[noend]{algpseudocode}

\usepackage[nottoc]{tocbibind}

\DeclareMathOperator*{\argmin}{arg\,min}

\usepackage[noend]{algpseudocode}
\usepackage{graphicx}
\graphicspath{ {./images/} }
\usepackage{tikz}

\newcommand{\R}{\mathbb{R}}
\newcommand{\Prob}{\mathbb{P}}
\newcommand{\TT}{\mathcal{T}}
\newcommand{\Train}{\TT_{\text{tr}}}
\newcommand{\Test}{\TT_{\text{te}}}
\newcommand{\XX}{\mathcal{X}} 
\newcommand{\YY}{\mathcal{Y}} 
\newcommand{\II}{\mathcal{I}} 
\newcommand{\Aa}{\mathcal{A}} 
\newcommand{\FF}{\mathcal{F}} 
\newcommand{\EE}{\mathbb{E}} 
\newcommand{\KK}{\mathcal{K}} 
 
\newcommand{\PP}{\mathcal{P}} 

\newcommand{\Itrain}{\II_{\text{tr}}}
\newcommand{\Itest}{\II_{\text{te}}}

\newcommand{\N}{\mathbb{N}} 

\newcommand{\err}{\mathcal{E}} 
\newcommand{\perr}{\mathcal{R}}
\newcommand{\perrest}{\widehat{\perr}}
\newcommand{\perrestsplit}{\perrest_{\text{split}}}

\newcommand{\cverr}[1]{\widehat{\err}^{\text{CV}}_{#1}}

\newcommand{\alg}[1]{\Aa_{#1}}

\newcommand{\rulefitted}[1]{\widehat{f}_{#1}}

\newcommand{\mspe}{\PP}
\newcommand{\ncvmse}{\widehat{\mspe}}

\newcommand{\ein}[1]{\widehat{\err}_{#1}}

\newcommand{\eout}[1]{\widehat{e}_{#1}}
\newcommand{\vout}[1]{\widehat{v}_{#1}}

\newcommand{\bias}{\mathcal{B}}
\newcommand{\biasest}{\widehat{\bias}}

\newcommand{\res}[1]{\widehat{r}_{#1}}

\newtheorem{lemma}{Lemma}

\newcommand{\Tr}{^T}
\newcommand{\inv}{^{-1}}
\newcommand{\norm}[1]{\left\lVert#1\right\rVert}
\newcommand{\abs}[1]{\left\lvert#1\right\rvert}

\title{Fast Computation of Nested Cross-Validation for Penalized Regression}

\author[1]{Shuyang (Henry) Cao\thanks{s72cao@uwaterloo.ca}}
\author[1]{Alex Stringer\thanks{alex.stringer@uwaterloo.ca}}
\affil[1]{Department of Statistics and Actuarial Science, Univeristy of Waterloo}
\date{\today}
\begin{document}
\maketitle

\begin{abstract}
    Cross-validation is a resampling procedure that provides a point estimate of generalization error for any predictive model. 
    Cross-validation is widely used for model selection and evaluation.
    Uncertainty in the cross-validation estimate is challenging to quantify, and estimation of its variance is known to require multiple runs of the entire resampling procedure.
    Nested cross-validation computes a prediction interval for the generalization error for a given model and training dataset by resampling the entire cross-validation procedure but incurs extraordinary computational cost. 
    We provide an efficient method for computing the nested cross-validation prediction interval using only a single model fit for some penalized regression models including ridge regression, spline smoothing, and some functional regression models. 
    We characterize when our proposed method should be expected to out-perform resampling-based nested cross-validation in various scaling regimes as well as in finite samples.
    Experiments for functional principle components regression demonstrate non-trivial cases in which our proposed method improves run times substantially.
\end{abstract}

\section{Introduction}\label{sec:intro}
Predictive models are built using an available set of training data with the goal of minimizing generalization error on future observations.
Predictions on test observations unavailable at the time of model training will incur error. 
Estimating and evaluating the uncertainty in estimates of this generalization error is fundamental to statistical machine learning.

Estimating and quantifying uncertainty in estimates of generalization error using only the training data used to fit the model is a substantial challenge.
Methods based on cross-validation and covariance penalties provide point estimates of the average generalization error across possible training sets,
and do not usually provide estimates of uncertainty.
This paper proposes a computationally feasible method for computation of the recently developed nested cross-validation score \citep{Bates2024} for quantifying the uncertainty
of predictions that are linear smooths of the inputs.
Applications include tuning parameter
selection for penalized regression models including ridge regression, spline smoothing, and functional principle component regression.

Figure \ref{fig:FPCR-pred-error} shows an example generalization error curve as a function of the penalization parameter in functional principle component
regression fit to a simulated dataset. The generalization error curve is fairly flat as a function of the tuning parameter, and the choice of this
parameter as the minimizer of this curve is subject to substantial uncertainty. When using ordinary or generalized cross-validation to estimate
this curve, an estimate of this uncertainty is not available.
Nested cross-validation provides an estimate of this uncertainty at a massive computational cost.
The method proposed in this paper computes this estimate of uncertainty at a greatly reduced computational cost.

Cross-validation estimates generalization error by resampling. 
The training data are partitioned, each partition is systematically removed, the model is fit to each remaining set of data, and then error is estimated
by a sample average of errors computed on the held out partitions.
Cross-validation applies very generally but requires
many model fits to obtain a point estimate of average generalization error.
Methods exist for quantifying uncertainty in the cross-validation estimate of average generalization error based on asymptotic normality and central limit theory \citep{dudoit2005,Austern,bayle2020}. 
These either do not provide computable estimation of the variance of the cross-validation score, or require the procedure to be run multiple times in order to estimate variance. 
\citet{bengio2004} proved that an unbiased estimator of variance of the cross-validation score cannot be obtained by running cross-validation only a single time, suggesting that it is necessary to run entire resampling procedure multiple times in order to quantify uncertainty in its output. 
The nested cross-validation procedure of \citet{Bates2024} provides prediction intervals for the 
generalization error in any predictive model by resampling the entire cross-validation
procedure. 
The nested cross-validation algorithm requires refitting the predictive model an extraordinary number of times, 
leading to intense computational burden even in simple problems. 

It is well-known that the ordinary cross-validation score in penalized regression problems 
can be computed from the residuals of a single model fit \citep[p.~257]{wood2017}, substantially reducing the computational cost. 
We propose to apply a similar idea to compute the nested cross-validation estimate and prediction interval using only a single model fit,
thereby substantially reducing its computational cost.
We derive a formula for the nested cross-validation score
and prediction interval that requires only a single model fit and some matrix algebra, leading to
$O(n)$ improvements in computation over the general algorithm when the number of folds is large. The formula applies to penalized regression models where the predictions are linear smooths of the inputs.

\section{Preliminaries}\label{sec:prelim}
\subsection{Predictive modeling and generalization error}\label{subsec:predictionerror}
Let $\XX\subseteq\R^{p}$ and $\YY\subseteq\R$. 
We observe training data $\TT = \left\{(X_1, Y_1),\ldots,(X_n,Y_n)\right\}$
consisting of inputs $X_i$ and outputs $Y_i$
where each $(X_i, Y_i)$ are distributed according to a joint distribution
$\Prob_{XY}$ on $\XX\times\YY$ independently for $i=1,\ldots,n$.
When appropriate we collect the samples from $\TT$ into $X\in\R^{n\times p}$ and $y\in\R^n$.
For any $\II\subseteq\{1,\ldots,n\}$ we let $X_{\II\II}$ and $X_{(-\II)(-\II)}$ denote the matrix
obtained by either including or excluding indices in $\II$ from matrix $X$, respectively. Note that $X_{\II\II}$ (or $X_{(-\II)(-\II)}$) is a submatrix by subsetting both rows and columns for indices in (or not in) $\II$. When subsetting only on rows or columns, we denote it as $X_{\II\cdot}$ or $X_{\cdot \II}$ respectively. A similar
definition applies for vector $y$,
and we let $\TT(\II) = \left\{(X_i, Y_i)\right\}_{i\in\II}$.
A prediction function is $f:\XX\to\YY$ from some space $\FF$, and a learning algorithm $\alg{\lambda}(\cdot):(\XX\times \YY)^{n}\to\FF$ takes
in training data and (optionally) a hyperparameter $\lambda\in\R^{d}$ and returns a prediction function $\alg{\lambda}(\TT) = \widehat{f}_\lambda\in\FF$.
We use the common shorthand $y = f(X)$ to indicate the $n$-dimensional vector $y = (f(X_{1,\cdot}),\ldots,f(X_{n,\cdot}))\Tr$.

Predictions $\widehat{y} = \widehat{f}_\lambda(X)$ returned by $\widehat{f}_\lambda$ will incur error compared to future observed data, and this is 
measured by a non-negative loss function $\ell:\YY\times \YY\to\R^{+}\cup\{0\}$ that satisfies 
$\ell(y,\widehat{y}) = 0$ if and only if $y = \widehat{y}$.
The definitions in the present section apply to any such loss function,
and the methods proposed in the remainder of the paper apply to the square loss, $\ell(y,\widehat{y}) = n\inv\norm{y-\widehat{y}}_2^2$.

The generalization error of a learning algorithm and training set is:
\begin{equation*}
    \perr(\alg{\lambda};\TT) = \EE\left[\ell\left\{\widehat{f}_\lambda(X^{*}), Y^{*}\right\}\Big\vert\TT\right].
\end{equation*}
Here $\widehat{f}_\lambda = \alg{\lambda}(\TT)$, and 
and $(X^{*},Y^{*})\sim\Prob_{XY}$ is a new observation not observed in $\TT$ but generated from the same joint distribution.
The goal of predictive modeling is to use the observed data $\TT$ to find some $\widehat{f}_\lambda\in\FF$ that incurs low generalization error.
The generalization error is a random variable that is a function of $\TT$.
The average generalization error is
\begin{equation*}
   \err(\alg{\lambda}) = \EE\left\{\perr\left(\alg{\lambda};\TT\right)\right\},
\end{equation*}
and is a fixed, constant number that depends on the learning algorithm $\alg{\lambda}$.
Both quantities depend on the
hyperparameter $\lambda$ through $\alg{\lambda}$.
The hyperparameter $\lambda$, if present, may be chosen by minimizing an estimate of $\err(\alg{\lambda})$
or $\perr(\alg{\lambda};\TT)$ over $\lambda$.

Estimation of $\err\{\alg{\lambda}\}$ and prediction of $\perr(\alg{\lambda};\TT)$ 
using only $\TT$ are both challenging tasks.
A desirable goal is a prediction interval for the random variable
$\perr(\alg{\lambda};\TT)$.
Many methods including cross-validation produce a point estimate
of $\err(\alg{\lambda})$ \citep{Bates2024}.

\subsection{Penalized regression}\label{subsec:penalizedregression}

Penalized regression defines $\alg{\lambda}(\TT) = \rulefitted{\lambda}$ where $\rulefitted{\lambda}(x) = x\Tr (X\Tr X + S_{\lambda})\inv X\Tr y$.
The matrix $S_{\lambda}\in\R^{p \times p}$ is called a penalty matrix and is model-specific.
The vector of predictions on the training set is a linear smooth of the training inputs, $\rulefitted{\lambda}(X) = H(\lambda)y$, with smoothing or ``hat'' matrix $H(\lambda) = X(X\Tr X + S_{\lambda})\inv X\Tr$.
This framework includes ridge regression, nonlinear regression via penalized splines \citep{hastie1990generalized}, and some functional and signal regression models including functional
principle components regression \citep{Reiss2007}.
In Section \ref{sec:FPCR-simulation},  we investigate the use of our proposed methods on hyperparameter selection for functional principle component regression models.

The hyperparameter $\lambda\in\R$ typically controls fit to the training data versus generalization to new data.
Higher values (say) can yield simpler fitted models with worse fit to the training data but a smaller gap between performance on the training versus
new data.
Lower values can yield complex models with excellent fit to the training data but a large gap between performance on the training versus new data.
Choosing a value of $\lambda$ that balances fit and performance is therefore a critical step in fitting a penalized regression model.
The presence of $\lambda$ defines an error curve since $\perr(\alg{\lambda};\TT)$ and $\err(\alg{\lambda})$ now depend on $\lambda$.
The value of $\lambda$ must be chosen and one option is to choose it to minimize an estimate of $\perr(\alg{\lambda};\TT)$ or $\err(\alg{\lambda})$.
These estimated error curves will have uncertainty associated with them.
Our proposed method reduces the computational
burden of uncertainty quantification of an estimated generalization error curve $\perrest(\alg{\lambda};\TT)$, enabling uncertainty quantification in practice
for hyperparameter selection in penalized regression.

\subsection{Estimation of expected generalization error via cross-validation}\label{subsec:crossvalidation}

Cross-validation is a resampling technique that estimates $\err(\alg{\lambda})$ by iteratively
applying $\Aa_\lambda$ to subsets of $\TT$ and then averaging the resulting estimates.
For some $K\in\N$ such that $n/K\in\N$, $K$-fold cross-validation defines a partition $\II_1,...,\II_K$ of $\{1,...,n\}$ and estimates $\err(\alg{\lambda})$ by
\begin{equation}\label{eqn:crossvalidationerror}
    \cverr{K}(\alg{\lambda}) = \frac{1}{n}\sum^n_{i=1}\ell\left\{Y_i, \widehat{f}^{-\KK(i)}_{\lambda}(X_i)\right\},
\end{equation}
where $\KK(i)\in\{1,\ldots,K\}$ denotes the index of the partition to which $y_i$ belongs and $\widehat{f}_{\lambda}^{-\KK(i)} = \Aa\{\TT(-\II_{\KK(i)});\lambda\}$.

Cross-validation requires fitting the model $K$ times and is generally computationally burdensome.
It is well-known that in penalized regression with training-set predictions $\widehat{y} = \widehat{f}_\lambda(X) = H(\lambda)y$ and square loss $\ell(y,\widehat{y}) = n\inv\norm{y-\widehat{y}}_2^2$, 
the $K$-fold cross-validation estimate can be computed using only a single model fit and some matrix algebra:
\begin{equation}\label{eqn:kfoldonestep}
    \cverr{K}(\alg{\lambda}) = \frac{1}{n}\sum^K_{k=1}(y_{\II_k} - \widehat{y}_{\II_k})\Tr(I - H(\lambda)_{\II_k\II_k})^{-2}(y_{\II_k} - \widehat{y}_{\II_k}).
\end{equation}
Eq. (\ref{eqn:kfoldonestep}) uses the well-known residual identity \citep[p.~257]{wood2017},
\begin{equation}\label{eqn:residualidentity}
y_{\II_j} - \widehat{f}_{\lambda, j}\left(X_{\II_j\cdot }\right) = \left(I - H(\lambda)_{\II_j \II_j} \right)^{-1}\left(y_{\II_j} - \widehat{y}_{\II_j} \right),
\end{equation}
which holds for penalized regression predictions. 
We show a proof of (\ref{eqn:residualidentity}) in Appendix \ref{appendix-proof-leave-several-out}. 

The case $K=n$ is called
leave-one-out cross-validation and has received special attention in the literature.
For penalized regression the one-step formula simplifies considerably:
\begin{equation*}
    \cverr{n}(\alg{\lambda}) = \frac{1}{n}\sum^n_{i=1}\left(\frac{y_i - x_i\Tr\widehat{\beta}}{1 - H(\lambda)_{ii}}\right)^2.
\end{equation*}
The diagonal element $H(\lambda)_{ii}$ is called the influence of $y_i$ on $\rulefitted{\lambda}$.

While cross-validation is highly general, imposing no restrictions on $\Aa_\lambda$ or $\Prob$,
as described in Section \ref{sec:intro}, it only returns a point estimate of $\err(\alg{\lambda})$. 
It is often practically desirable to predict $\perr\{\alg{\lambda};\TT\}$, the error incurred
by fitting the model to the data that were observed, rather than estimate $\err(\alg{\lambda})$ which represents the average performance of $\Aa_\lambda$ across data that might have been observed. These two tasks are substantively different.
In either case, uncertainty must be quantified.

\subsection{Prediction intervals for generalization error using nested cross-validation}\label{subsec:nestedcrossvalidation}

Motivated by the proven lack of ability to quantify uncertainty in $\cverr{K}$ using only
a single run of cross-validation, \citet{Bates2024} introduce nested cross-validation.
Nested cross-validation
computes a prediction interval for $\perr\{\Aa_\lambda;\TT\}$ by resampling
the entire cross-validation procedure.
This involves running cross-validation many times,
with each run of cross-validation requiring many model fits.
The result is uncertainty quantification for an estimate
of generalization error $\perr\{\alg{\lambda};\TT\}$, at substantial computational cost.

Fundamental to cross-validation and its nested extension is the notion of a train-test split. 
Let $\Itrain \cup \Itest = \{1,\ldots,n\}$ and $\Itrain\cap\Itest=\emptyset$. Then let $\Train = (X_i, Y_i)_{i \in \Itrain}$ and $\Test = (X_i, Y_i)_{i \in \Itest}$ be a partition of $\TT$
into a training and test set, and let $\widehat{f}_\lambda = \Aa_\lambda(\Train)$.
Then
\begin{equation}\label{eqn:testerror}
\perrestsplit = \frac{1}{\abs{\Test}}\sum_{i\in\Itest}\ell\left\{ Y_{i}, \widehat{f}_\lambda(X_i)\right\}
\end{equation}
satisfies $\EE(\perrestsplit | \Train) = \perr\{\Aa_\lambda;\Train\}$.
Let $\perrest(\Train)$ be any prediction of the generalization error $\perr\{\Aa_\lambda;\Train\}$ that depends on $\Train$ but not $\Test$.
Nested cross-validation decomposes the mean squared generalization error on the training set,
\begin{equation}\label{eqn:mspe1}
\mspe(\Aa_\lambda;\Train) = \EE\left[\left(\perrest(\Train) - \perr\{\Aa_\lambda;\Train\}\right)^2\right],
\end{equation}
as
\begin{equation}\label{eqn:mspe2}
\mspe(\Aa_\lambda;\Train) = \EE\left[\left( \perrest(\Train) - \perrestsplit \right)^2\right] - \EE\left[\left( \perrestsplit - \perr\{\Aa_\lambda;\Train\}\right)^2\right]
\end{equation}
\citep[Lemma 1]{Bates2024}. 
Nested cross-validation is the estimation of (\ref{eqn:mspe1}) by
applying cross-validation to further resampled datasets to estimate the two terms on the right-hand side of (\ref{eqn:mspe2}), using (\ref{holdout}) and (\ref{eqn:ncvmse}) below.

Partition the original data into $K\in\N$ folds $\TT(\II_1),\ldots,\TT(\II_K)$ as in cross-validation (Section \ref{subsec:crossvalidation}). 
For $j=1,\ldots,K$, run a $(K-1)$-fold cross-validation on data $\TT(-\II_{j}) = \TT(\cup_{i\neq j}\II_i)$ and
obtain the error estimate $\ein{j} = \cverr{K-1}[\Aa_\lambda\{\TT(-\II_{j})\}]$ via Equation (\ref{eqn:crossvalidationerror}).
For each $j$, further compute $\widehat{f}_{\lambda,j} = \Aa_\lambda\{\TT(-\II_{j})\}$ which leads to holdout/test error estimates:
\begin{equation}\label{holdout}
\begin{aligned}
\eout{j} &= \frac{1}{\abs{\II_{j}}}\sum_{i\in\II_j}\ell\left\{Y_i, \widehat{f}_{\lambda,j}(X_i)\right\}, \\
\vout{j} &= \frac{1}{\abs{\II_{j}}-1}\sum_{i\in\II_j}\left[\ell\left\{Y_i, \widehat{f}_{\lambda,j}(X_i)\right\} - \eout{j}\right]^2.
\end{aligned}
\end{equation}
Each $\eout{j}$ is an unbiased estimate of $\perr[\Aa_\lambda;\TT(-\II_{j})]$ and $\vout{j}$ is an unbiased estimate of the variance of $\eout{j}$.

Nested cross-validation estimates $\perr\{\Aa_\lambda;\Train\}$ by $\perrest\{\Aa_\lambda;\Train\} = K\inv (\ein{1} + \cdots + \ein{K})$
and $\mspe(\Aa_\lambda;\TT)$ by
\begin{equation}\label{eqn:ncvmse}
\ncvmse(\Aa_\lambda;\TT) = \frac{1}{K}\sum_{j=1}^{K}\left[ \left\{ \ein{j} - \eout{j} \right\}^2 \right] - \frac{\vout{j}}{\abs{\II_j}}.
\end{equation}
Along with an appropriate
bias correction term $\biasest$, the nested cross-validation $(1-\alpha)100\%$ prediction interval for $\perr\{\Aa_\lambda;\Train\}$ is therefore
\begin{equation}\label{eqn:ncvpredint}
\perrest\{\Aa_\lambda;\Train\} - \biasest \pm z_{1-\alpha/2}\ncvmse(\Aa_\lambda;\TT)^{1/2},
\end{equation}
where $z_{p}$ is the $p$-th percentile of the standard Normal distribution.

\citet{Bates2024} further suggest that the estimation of mean squared generalization error may be
improved by repeating the entire nested cross-validation procedure some number $L\in\N$ times
and using the average of the $\ncvmse(\Aa;\TT)$ obtained from each run to form the prediction
interval (\ref{eqn:ncvpredint}). 
They suggest $L = 200$ although limited explanation is provided.

Where $K$-fold cross-validation already incurs a computational cost that is $K$ times that of
fitting the model, nested cross-validation further exacerbates this computational burden by
requiring $K-1$ additional model fits per fold and then repeating the entire procedure $L$ times.
The total complexity after removing some redundant computations is $LK(K+1)/2$ model fits. 
This represents an extraordinary computational burden,
rendering nested cross-validation impractical even in simple models. 
For example, running $5$-fold nested cross-validation using the suggested $L = 200$ requires $3000$ model fits.
Further, when the goal of 
predicting generalization error is to do so as a function of the hyperparameter $\lambda$ and then
find the value $\widehat{\lambda}$ that minimizes the estimate, the entire nested cross-validation
procedure will have to be repeated for each value of $\lambda$ that is tried. 
Our proposed method substantially reduces this computational burden by computing
the nested cross-validation prediction interval using a single model fit and some matrix algebra, making application of nested cross-validation feasible for penalized regression.

\section{Fast computation of nested cross-validation prediction intervals in penalized regression models}
\label{sec:fast-ncv}
Our main contribution is to provide a closed-form formula of computing nested cross-validation for penalized regression models. 
We give the formula in Proposition \ref{ncvformula}. 
Let $M_{n,K} \in \R^{\frac{2n}{K} \times \frac{2n}{K}}$ be the matrix obtained by setting the first $n/K$ diagonal entries of the identity matrix of dimension $2n/K$ to be zero.
\begin{proposition}
\label{ncvformula}
   Fix $L,L_B\in\N$. For $l=1,\cdots,L$, we denote the sampled folds as $\II_1^l, \cdots, \II^l_K$. Similarly, we denote folds $\II_1^{l_b}, \cdots, \II_{K}^{l_b}$ for $l_b = 1, \cdots, L_B$.
 The nested cross-validation point estimate of generalization error is
 \begin{equation}
    \begin{aligned}
        \perrest\{\Aa_\lambda;\Train\} &= \frac{1}{Ln(K-1)}\sum_{l=1}^L\sum_{(k, j), k < j}\left(y_{I_{jk}^l} - \widehat{y}_{\II^l_{jk}}\right)^T\left(I - H(\lambda)_{\II^l_{jk}\II^l_{jk}}\right)^{-2}\left(y_{I^l_{jk}} - \widehat{y}_{\II^l_{jk}}\right), \label{eqn:ncvpeformula}
    \end{aligned}
\end{equation}
and the bias correction term is
\begin{equation}
    \label{eqn:ncv-bias-formula}
 \biasest=   \left(1 + \frac{K-2}{K}\right)\left( \perrest\{\Aa_\lambda;\Train\} - \frac{1}{nL_B}\sum^{L_B}_{l_{b}=1}\sum_{j=1}^K\left(y_{\II_{j}^{l_b} } - \widehat{y}_{\II_{j}^{l_b}}\right)^T\left(I - H(\lambda)_{\II_j^{l_b} \II_j^{l_b}}\right)^{-2} \left(y_{\II_{j}^{l_b} } - \widehat{y}_{\II_{j}^{l_b}}\right)\right).
\end{equation}
The nested cross-validation estimate of the mean-squared error is
\begin{equation}
\begin{aligned}
    \ncvmse(\Aa_\lambda;\TT) &= \frac{1}{KL}\sum^L_{l=1}\sum^K_{j=1}\Biggl\{\\
    &\biggl[\frac{K}{n(K-1)}\sum_{k \neq j}\left(y_{\II^l_{jk}} - \widehat{y}_{\II^l_{jk}}\right)^T\left(I - H(\lambda)_{\II^l_{jk}\II^l_{jk}}\right)^{-1} M_{n,K}\left(I - H(\lambda)_{\II^l_{jk}\II^l_{jk}}\right)^{-1}\left(y_{\II^l_{jk}} - \widehat{y}_{\II^l_{jk}}\right) - \\& \frac{K}{n}\left(y_{\II^l_j} - \widehat{y}_{\II^l_j} \right)^T\left(I - H(\lambda)_{\II^l_j \II^l_j} \right)^{-2}\left(y_{\II^l_j} - \widehat{y}_{\II^l_j} \right)\biggr]^2 - \quad \\ & 
    \frac{\frac{1}{|\II^l_j| - 1}\sum_{v}\left(\left( \left(I - H(\lambda)_{\II^l_j \II^l_j} \right)^{-1}\left(y_{\II^l_j} - \widehat{y}_{\II^l_j} \right) \right)^2_v - \frac{1}{|\II^l_j|}\sum_{u}\left( \left(I - H(\lambda)_{\II^l_j \II^l_j} \right)^{-1}\left(y_{\II^l_j} - \widehat{y}_{\II^l_j} \right) \right)^2_u\right)^2}{|\II^l_j|}\Biggr\}, \label{eqn:ncvmseformula}
\end{aligned}
\end{equation}
\end{proposition}
The formulas in Proposition \ref{ncvformula} compute the entire prediction interval (\ref{eqn:ncvpredint}) using only a single model fit to $\TT$. 
An investigation of the computational complexity of this formula and the associated implementation details are given in Section \ref{sec-complexity}. 
\begin{proof}[Proof of Proposition \ref{ncvformula}]
Fix arbitrary $l\in\{1, \cdots, L\}$. 
Apply Eq. (\ref{eqn:residualidentity}) to fold $\II^l_k$ and pair of folds $\II^l_k \cup \II^l_j$ with $k \neq j$ to obtain 
\begin{align*}
    \res{j} &\equiv y_{\II^l_j} - \widehat{f}_{\lambda, j}\left(X_{\II^l_j\cdot }\right) = \left(I - H(\lambda)_{\II^l_j \II^l_j} \right)^{-1}\left(y_{\II^l_j} - \widehat{y}_{\II^l_j} \right),\\
        \res{jk} &=y_{\II^l_j \cup \II^l_k}  - \widehat{f}_{\lambda, jk}\left(X_{\II^l_j \cup \II^l_k \cdot}\right) =  \left(I - H(\lambda)_{\II^l_{jk}\II^l_{jk}}\right)^{-1}\left(y_{\II^l_{jk}} - \widehat{y}_{\II^l_{jk}}\right). 
\end{align*}
For each $j \in \{1,\ldots, K\}$, by (\ref{holdout}) we have
\begin{equation}
    \begin{aligned}
        \eout{j} = \frac{K}{n}\left(y_{\II_j} - \widehat{y}_{\II_j} \right)\left(I - H(\lambda)_{\II_j \II_j} \right)^{-2}\left(y_{\II_j} - \widehat{y}_{\II_j} \right). \label{eqn:eoutformula}
    \end{aligned}
\end{equation}
Recall in Section \ref{subsec:nestedcrossvalidation}, $\ein{j}$ is the result of performing a $(K-1)$-fold CV on data $\TT(\cup_{i\neq j}\II_i)$. 
Now define \begin{align*}
    \res{jk^\prime} = M_{n, K}\res{jk}
\end{align*}
where $M_{n, K} \in \R^{\frac{2n}{K} \times \frac{2n}{K}}$ is defined as $M_{n,K} = \text{block.diag}(0_{n/K\times n/K}, I_{n/K})$.

It follows that 
\begin{equation}
\begin{aligned}
    \ein{j} &= \frac{K}{n(K-1)} \sum_{k \neq j}\left(\res{jk^\prime}\right)^T \res{jk^\prime}\\
    &=\frac{K}{n(K-1)} \sum_{k \neq j}\left(y_{I_{jk}^l} - \widehat{y}_{\II_{jk}^l}\right)^T\left(I - H(\lambda)_{\II_{jk}^l\II_{jk}^l}\right)^{-1} M_{n,K}\left(I - H(\lambda)_{\II_{jk}^l\II_{jk}^l}\right)^{-1}\left(y_{I_{jk}^l} - \widehat{y}_{\II_{jk}^l}\right). \label{eqn:einformula}
\end{aligned}
\end{equation}
Summing these residuals gives:
$$
\sum_{k=1}^K\ein{k} = \sum_{k = 1}^K \sum_{j=1}^{k-1} \left(y_{I_{jk}^l} - \widehat{y}_{\II^l_{jk}}\right)^T\left(I - H(\lambda)_{\II^l_{jk}\II^l_{jk}}\right)^{-2}\left(y_{I^l_{jk}} - \widehat{y}_{\II^l_{jk}}\right). 
$$
The point estimate (\ref{eqn:ncvpeformula}) follows.
Eq. (\ref{eqn:ncvmseformula}) is obtained by combining (\ref{eqn:eoutformula}),(\ref{eqn:einformula}) and (\ref{eqn:ncvmse}) and using the definition of $\vout{j}$
as the sample variance of $\res{j}$.

\citet[Appendix~C]{Bates2024} proposed a bias correction term as 
\begin{equation}
    \label{eqn:bias-correctiion}
    \biasest = \left(1 + \frac{K-2}{K}\right)\left(\perrest\{\Aa_\lambda;\Train\} -  \cverr{K, L_B}\right), 
\end{equation}
where $\cverr{K, L_B}$ denotes the average CV score of repetitively $K$-fold CV for $L_B$ times. 
Applying (\ref{eqn:residualidentity}) gives \begin{equation}
    \label{eqn:CV-score-repeate}
    \cverr{K, L_B} = \frac{1}{nL_B}\sum^{L_B}_{l_{b}=1}\sum_{j=1}^K\left(y_{\II_{j}^{l_b} } - \widehat{y}_{\II_{j}^{l_b}}\right)^T\left(I - H(\lambda)_{\II_j^{l_b} \II_j^{l_b}}\right)^{-2} \left(y_{\II_{j}^{l_b} } - \widehat{y}_{\II_{j}^{l_b}}\right). 
\end{equation} Then (\ref{eqn:ncv-bias-formula}) follows from (\ref{eqn:bias-correctiion}), (\ref{eqn:CV-score-repeate}) and (\ref{eqn:ncvpeformula}). (\ref{eqn:ncvmseformula}), (\ref{eqn:ncvpeformula}) and (\ref{eqn:ncv-bias-formula}) together give a closed-form formula of (\ref{eqn:ncvpredint}) as desired. 
\end{proof}

\section{Computational complexity}
\label{sec-complexity}
The proposed one-step nested cross-validation formula in Proposition \ref{ncvformula} may be applied in any penalized regression model
(Section \ref{subsec:penalizedregression}). 
The one-step formula depends on matrix operations which themselves may become computationally burdensome for larger data and more complex models.
We aim to quantify if and by how much the one-step formula improves the computational complexity of nested cross-validation under each combination of
the common asymptotic regimes of $K$ fixed or $K\in O(n)$, and $p$ fixed or $p \in O(n^\alpha)$ for $0 < \alpha \leq 1$. 
The fixed-$K$ case is called $K$-fold (nested) cross-validation and the $K \in O(n)$ case is called leave-$k$-out (nested) cross-validation, where $k = n/K$.
The fixed-$p$ regime is called low-dimensional regression and the $p\in O(n^\alpha)$ regime is called high-dimensional regression.
We find that the one-step formula yields an $O(n)$ improvement over resampling for leave-$k$-out nested cross-validation for both
low- and high-dimensional regression.

When fitting nested cross-validation to a specific dataset, asymptotic scaling is less relevant than actual runtime performance.
It is convenient to understand for which values of $n$, $p$, and $K$ the one-step formula should be faster than model refitting.
We address both of these questions by accounting for the floating point operations required to implement both procedures.
Throughout, we assume that the penalized regression model is implemented such that it requires the standard $O(np^2 + p^3)$ floating point operations to fit,
which correspond to the singular value decomposition of the $n\times p$ model matrix and solving $p\times p$ linear systems.
We also recall the Woodbury Identity,
\begin{equation}\label{eqn:residualwoodbury}
    \left(I - H(\lambda)_{\II_{jk}\II_{jk}}\right)^{-1} = I + X_{\II_{jk}\cdot}\left(X^T_{(-\II_{jk})\cdot}X_{(-\II_{jk})\cdot} + \lambda S\right)^{-1}X_{\II_{jk}}^T,
\end{equation}
which may improve the computation of residuals depending on the dimensions of the matrices involved.

Table \ref{tab:computationalcomplexity} compares the asymptotic computational complexity required to compute (\ref{eqn:ncvmse}) by
refitting the model and by implementing Proposition \ref{ncvformula} with and without (\ref{eqn:residualwoodbury}).
Implementation of Proposition \ref{ncvformula} with (\ref{eqn:residualwoodbury}) is never asymptotically worse than refitting model. 
When $K \in O(n)$, in both the high-dimensional $p \in O(n^{\alpha})$ 
and low-dimensional $p$-fixed regimes, the one-step formula without (\ref{eqn:residualwoodbury}) reduces the computational complexity by order $O(n)$ and is therefore asymptotically more
computationally efficient than refitting the model.
When $K$ is fixed, the computationally complexity of Proposition \ref{ncvformula} is of the same
asymptotic order as refitting the model.
The empirical evaluation shown in Figure \ref{fig:time_compare_pfixed} and Figure \ref{fig:time_compare_npscale}
demonstrates cases where a substantial speedup is observed in practice.

\begin{table}[htbp]
\centering
\caption{Asymptotic numbers of floating point operations required to compute the nested cross-validation prediction interval by model refitting (Eq. (\ref{eqn:ncvmse}) and using Proposition \ref{ncvformula} with and without the Woodbury identity (Eq. (\ref{eqn:residualwoodbury})). 
Shown are the four asymptotic regimes: $K$-fold ($K$ fixed) and leave-$k$-out ($K \in O(n)$) nested cross-validation, as well as low-dimensional ($p$ fixed) and high-dimensional ($p \in O(n^\alpha)$ for $0 < \alpha \leq 1$) penalized regression.}

\label{tab:computationalcomplexity}
\begin{tabular}{l|lll}
Regime & Refitting Model & Prop. \ref{ncvformula} with (\ref{eqn:residualwoodbury})  & Prop. \ref{ncvformula} without (\ref{eqn:residualwoodbury}) \\
\hline
$K$, $p$ fixed & $O\left(K^2np^2 + K^2p^3\right)$ & $O\left(K^2np^2 + K^2p^3\right)$ & $O\left(n^3/K + n^2p + Knp^2 \right)$ \\
$K \in O(n)$, $p$ fixed & $O\left(n^3p^2\right)$ & $O\left(n^3p^2\right)$ & $O\left(n^2p^2\right)$ \\
$K$ fixed, $p \in O(n^\alpha)$ & $O\left(K^2n^{1 + 2\alpha} \right)$ & $O\left(K^2n^{1 + 2\alpha} \right)$ & $O(n^3/K + n^{2 + \alpha} + Kn^{1+2\alpha})$ \\
$K \in O(n)$, $p \in O(n^\alpha)$ & $O\left(n^{3 + 2 \alpha}\right)$ &$O\left(n^{3 + 2 \alpha}\right)$ & $O\left(n^{2 + 2\alpha}\right)$ \\
\bottomrule
\end{tabular}
\end{table}

Given a finite set of data and a predictive model,
the actual values of $n$, $p$, and $K$ determine which of the three methods compared in Table \ref{tab:computationalcomplexity} should be applied.
The same accounting of floating point operations that gives Table \ref{tab:computationalcomplexity} yields the following result.
\begin{proposition}\label{prop:finitesample}
Fix $n$ and $p$.
The one-step formula for the nested cross-validation prediction interval given in Proposition \ref{ncvformula}
requires less floating point operations than refitting the model whenever $K > K^{*}$, where $K^{*}$ is the smallest root in $[2,n]$
of the quartic polynomial $f(K) = \alpha^{(0)}_{n,p} + \alpha^{(1)}_{n,p}K + \alpha^{(2)}_{n,p}K^2 + \alpha^{(3)}_{n,p}K^3 + \alpha^{(4)}_{n,p}K^4$
with coefficients
\begin{align*}
\alpha^{(0)}_{n,p} &= 6n^3, \\
\alpha^{(1)}_{n,p} &= -8n^3 + 12pn^2 + 30n^2, \\
\alpha^{(2)}_{n,p} &= -24pn^2 - 60n^2, \\
\alpha^{(3)}_{n,p} &= -18p^2n + p^3 + 15p^2 + 6n + 6pn, \\
\alpha^{(4)}_{n,p} &= 6p^2n + p^3 + 15p^2 + 6pn.
\end{align*}
\end{proposition}
Given a choice of $K$, a dataset of size $n$, and a model of size $p$, 
Proposition \ref{prop:finitesample} is used in practice to decide whether to perform nested cross-validation by refitting the model or by using the one-step formula. 
While it appears to require finding the root(s) of a quartic polynomial, because the only acceptable values of $K$ are those for which $n/K\in\N$,
we simply compute $f(K)$ for all such divisors of $n$ and find the point where $f(K)$ crosses zero. 
We provide justification of Proposition \ref{prop:finitesample} in Appendix \ref{appendix-proof-prop-2}. 
Figure \ref{fig:cutoff} reports measured relative computational times between Proposition \ref{ncvformula} and model refitting in which the value of $K^{*}$ correctly predicts the point at which Proposition \ref{ncvformula} becomes faster.

\begin{figure}[ht]
    \centering
    \includegraphics[width=16cm]{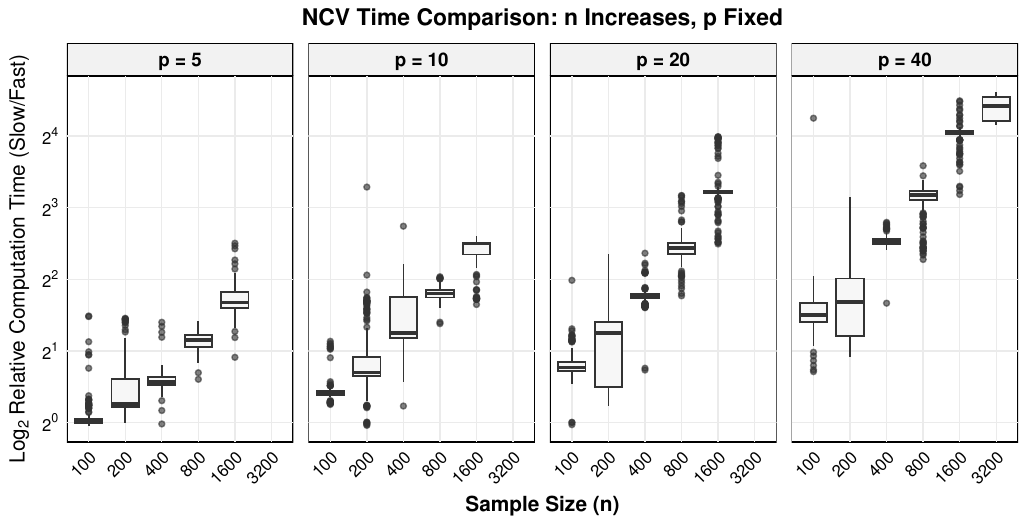}
    \caption{Change of $\log_2$ scale relative computation time for fitting a cubic spline. Relative computation means the computation time of refitting the model (slow) over the computation time of nested cross-validation formula (fast). Regime corresponds to $K\in  O(n)$, $p$ fixed in Table \ref{tab:computationalcomplexity}, meaning $n$ increases and $p$ keeps fixed. }
        \label{fig:time_compare_pfixed}
\end{figure}

\begin{figure}[ht]
    \centering
    \includegraphics[width=16cm]{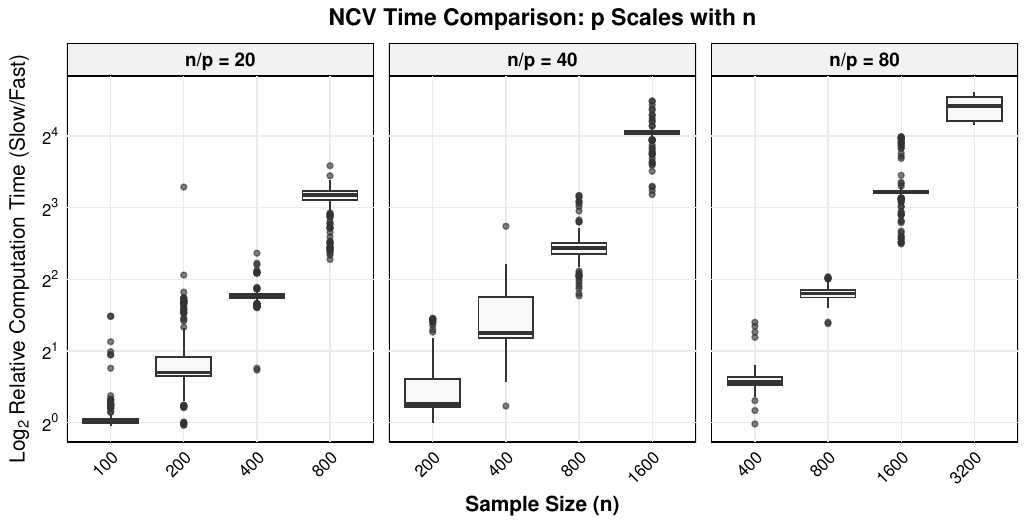}
        \caption{Change of log scale relative computation time for fitting a cubic spline. Relative computation means the computation time of refitting the model (slow) over the computation time of nested cross-validation formula (fast). Regime corresponds to $K\in O(n)$, $p \in O(n^\alpha)$ in Table \ref{tab:computationalcomplexity}, meaning $n$ increases and $p$ also scales with $n$. }
        \label{fig:time_compare_npscale}
\end{figure}

\clearpage
\begin{figure}[ht]
    \centering
    \includegraphics[width=16.5cm]{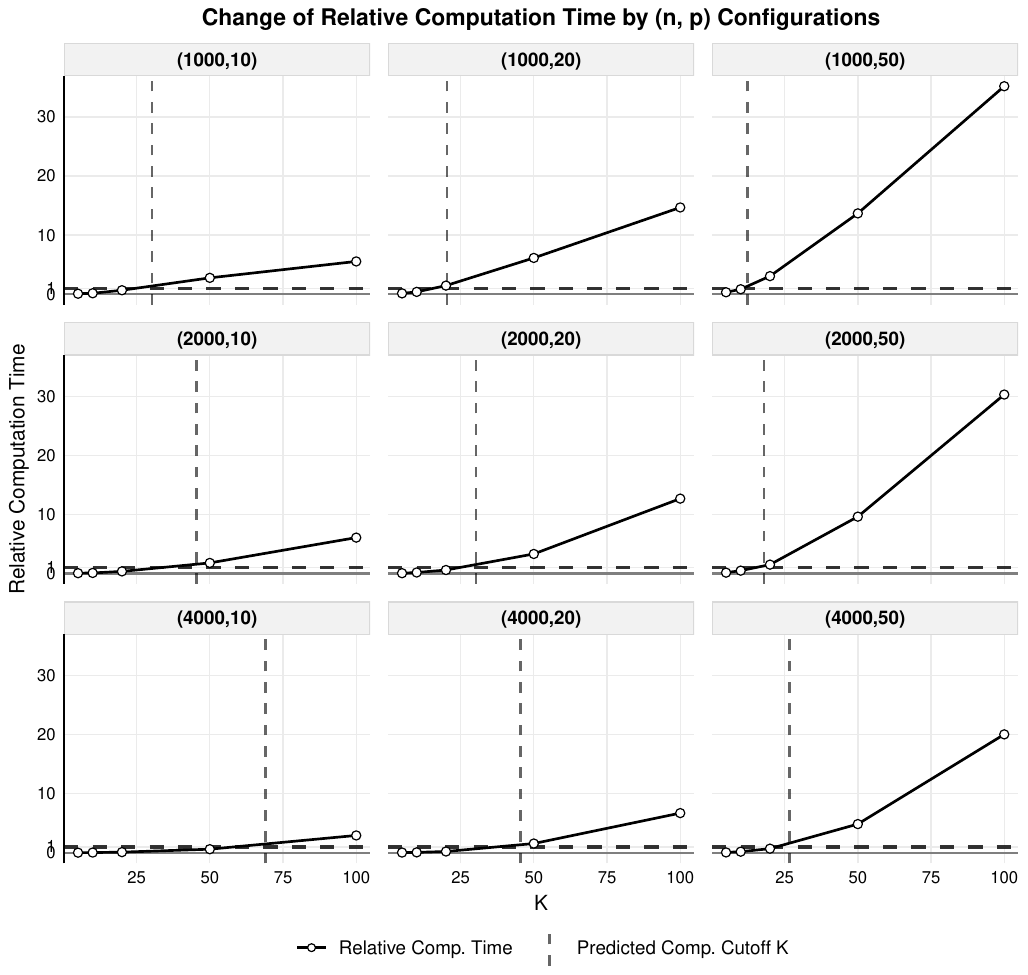}
    \caption{Change of relative computation time (refitting model / nested cross-validation formula in Proposition \ref{ncvformula}) with the increase of $K$ and varying pair of $(n, p)$. Horizontal dotted line is $1$, meaning computation time of both is equal. Vertical dotted line is the predicted $K^*$. Relative computation time is expected to rise above $1$ when $K$ is bigger than predicted $K^*$.}
    \label{fig:cutoff}
\end{figure}

\section{Application to tuning parameter selection in functional principle components regression}
\label{sec:FPCR-simulation}
\subsection{Signal regression and functional principal components}
\label{sec:FPCR-model-setup}
Let $y_i = \beta_0 + \int_0^1x_i(t)\omega(t)dt + \epsilon_i, i = 1,\ldots,n$, with $\epsilon_i\sim\text{N}(0,\sigma^2)$ independently.
This signal regression model is an example of a scalar-on-function regression model popular in the functional data analysis literature.
The unknown functional parameter is $\omega:\R\to\R$.
The covariates $x_i(t)$ are observed at fixed time points $t_1,\ldots,t_N$, where $N$ may be very large compared to $n$.
Approximation of the integral using a Riemann sum combined with a B-spline expansion $\omega(t) \approx b_1(t)\beta_1 + \cdots + b_M(t)\beta_M$ for some $M\in\N$ leads to the discrete form of the model:
\begin{equation}
    \label{model:signal_reg_basis}
    y_i = \beta_0 + \sum^N_{j=1}\Delta_j x_{i}(t_j)\left\{\sum^M_{m=1}b_{m}(t_j)\beta_m\right\} + \epsilon_i, i = 1,\ldots, n,
\end{equation}
where $\Delta_j = |t_{j+1} - t_j| \in \R$ for $2 \leq j \leq N$.
For further details on signal regression and functional data analysis, see \citet{ramsay2006}.
The unknown parameters are the intercept, $\beta_0$, and the spline weights, $\beta = (\beta_1,\ldots,\beta_M)\Tr$.
To estimate the spline weights it is desirable to minimize the penalized least squares criterion:
$$
\widehat{\beta}(\lambda) = \argmin_{\beta\in\R^M}\sum_{i=1}^{n}\left[y_i - \beta_0 - \sum^N_{j=1}\Delta_jx_{i}(t_j)\left\{\sum^M_{m=1}b_{m}(t_j)\beta_m\right\}\right] + \lambda\int_0^1\{\omega^{\prime\prime}(t)\}^2dt.
$$
Define $S\in\R^{M\times M}$ by $S_{km} = \int_0^1 b_k^{\prime\prime}(t)b_m^{\prime\prime}(t)dt$, 
$X\in\R^{n\times N}$ by $X_{ij} = \Delta_jx_i(t_j)$, and $B\in\R^{N\times M}$ by $B_{jm} = b_m(t_j)$.
The penalized least squares criterion is written in vector form as
$$
\widehat{\beta}(\lambda) = \argmin_{\beta\in\R^M}\norm{y - \beta_01_n - XB\beta}_2^2 + \lambda\beta\Tr S\beta,
$$
with unique solution
$$
\widehat{\beta}(\lambda) = \left((XB)^TXB + \lambda S\right)^{-1}(XB)^Ty.
$$
The corresponding training set predictions are:
$$
\widehat{y}(\lambda) = XB \left((XB)^TXB + \lambda S\right)^{-1}(XB)^Ty \equiv H(\lambda)y. 
$$
The predictions are obtained as a linear smooth of the training inputs, and
our proposed method for computing the nested cross-validation prediction interval applies.

When the signal is densely observed, the model may be substantially over-parameterized.
Parsimony is achieved by performing the singular value decomposition $XB = UDV\Tr$ with $U\in\R^{n\times M}$ a column-orthonormal matrix of left singular vectors, $D = \text{diag}(\delta_1,\ldots,\delta_M)$ a diagonal matrix of ordered singular values, and $V\in\R^{M\times M}$ a square orthonormal matrix of right singular vectors.
Some of the singular values may be close to zero, and we retain some number $p \leq M$ of them such that $(\delta_1^2 + \cdots +\delta_p^2) / (\delta_1^2 + \cdots +\delta_M^2)$ is close to $1$.
Denote by $V_p$ then first $p$ columns of $V$.
The functional principle component predictions are
$$
\widehat{y}_p(\lambda) = (XBV_p)\left[(XBV_p)^T XBV_p + \lambda V_p SV_p\right]^{-1}(XBV_p)^Ty.
$$
This model is called Functional Principal Component Regression (FPCR) \citep{Reiss2007}. 
The predictions are still written in the form $\widehat{y}_p(\lambda) = H_p(\lambda)y$ and so are compatible with the one-step formula of Proposition \ref{ncvformula}.
Accordingly, we proceed with nested cross-validation to estimate and quantify uncertainty in the generalization error $\perr\{\alg{\lambda};\TT\}$.

\subsection{Selection of $\lambda$}
We perform a simulation study to examine the uncertainty in cross-validation-based assessment of generalization error functional principle components regression and its
impact on hyperparameter selection.
This is made feasible by the nested cross-validation formula we provide in Proposition \ref{ncvformula}.
Our experimental setup is motivated \citet{Reiss2009}, in which they compare the performance of smoothing parameter between restricted maximum likelihood and generalized cross-validation. 
Details of the restricted maximum likelihood approach are given in Appendix \ref{appendix-REML}. 
We compare hyperparameter selection for functional principle components regression
using restricted maximum likelihood, generalized cross-validation, and leave-one-out cross-validation.
We denote the selected smoothing parameter $\lambda$ by these methods as $\lambda_{\text{REML}}, \lambda_{\text{GCV}}$ and $\lambda_{\text{LOO}}$ separately. 

Let $N=n=100$. The data generation procedure is described sequentially, following the same order in which the variables are generated in the simulation. A vector of time sites is generated by truncating the interval $[0, 1]$, namely $t = (0.01, 0.02, \cdots, 1)^T$ and the constant quadrature weight is $\Delta_t = 0.01$. Next, we generate a covariance matrix $\Sigma \in \R^{100 \times 100}$ as $\Sigma_{ij} = \exp(-|t_i - t_j|/2)$. Then for each row $X_i^T$ of the design matrix $X$, we generate IID $X_i \sim N(0, \Sigma)$. The signal $\omega$ is generated as $\omega = \sin(t) = (\sin(0.01\cdot 2\pi), \sin(0.02\cdot 2\pi), \cdots, \sin(1\cdot 2\pi))^T$. Lastly, we generate $y$ as $y = \Delta_t X\omega + \epsilon$ with IID $\epsilon_i \sim N(0, 0.1)$ and $\epsilon = (\epsilon_1, \cdots, \epsilon_{100})$. When fitting the FPCR, we absorb $\Delta_t$ and treat the design matrix as $\Delta_t X$. We choose $M = 40$ as the number of spline basis expansion, hence $B \in \R^{100 \times 40}$ and $S \in \R^{40 \times 40}$. By varying $p \in \{5, 10\}$, we obtain different proportions of variance. 

Figure \ref{fig:FPCR} shows the true generalization error, cross-validation estimate and nested cross-validation prediction intervals, restricted maximum likelihood functions, and the selection of $\lambda$ by each method.
Two choices of proportions of variance explained are shown.
The bias of the cross-validation score improves with a larger number of principle components.
The nested cross-validation prediction interval accounts for this bias and broadly contains the observed generalization errors.
The estimated generalization error is very flat as a function of $\lambda$, 
and the estimated uncertainty is high. 
The generalization error increases when $\lambda$ is too large, but there is little
to distinguish smaller values of $\lambda$ in terms of this score.
This is not true of restricted maximum likelihood, a result consistent with the
conclusions of \citet{Reiss2009} for functional principle components regression as well as
\citet{wood2011fast} for generalized additive models.
Availability of estimated uncertainty in the estimation of generalization error makes this
result more transparent.

\begin{figure}[ht]
  \centering
  \begin{subfigure}[b]{0.49\textwidth}
    \centering
    \includegraphics[width=\linewidth]{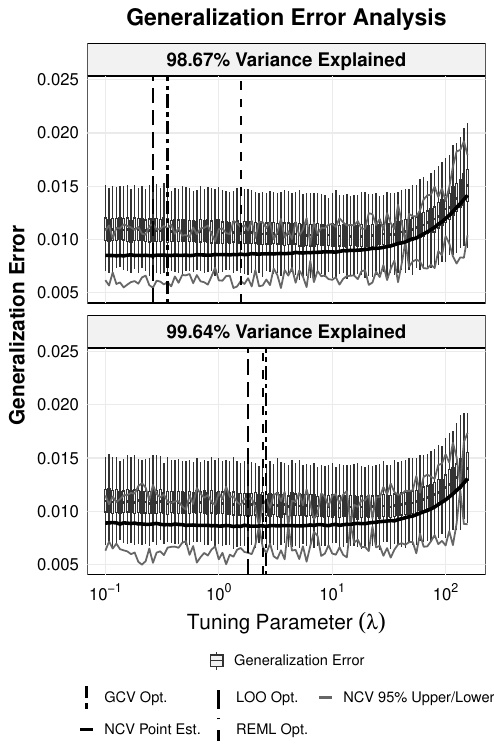}
    \caption{Generalization error with Tuning Parameter}
    \label{fig:FPCR-pred-error}
  \end{subfigure}%
  \hfill
  \begin{subfigure}[b]{0.49\textwidth}
    \centering
    \includegraphics[width=\linewidth]{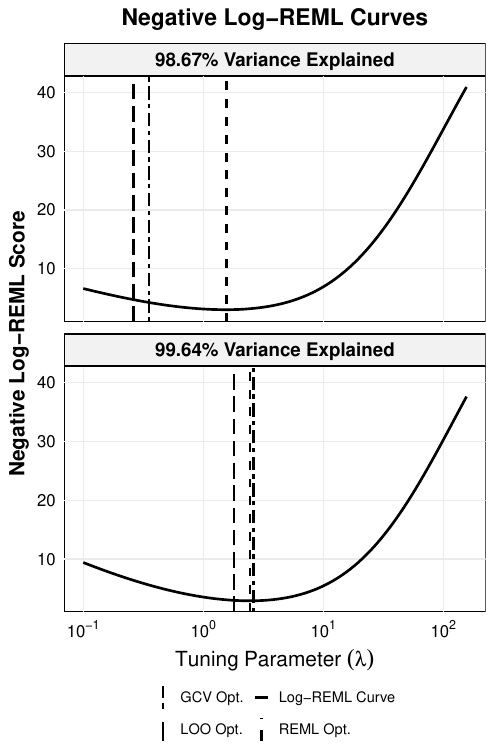}
    \caption{Negative Log-REML Curves}
    \label{fig:reml_curves}
  \end{subfigure}
  \caption{(a) is the boxplot of generalization errors with results of tuning parameters. nested cross-validation point estimates with prediction interval are shown as solid lines. Choices of tuning parameters that give the lowest GCV score, LOO score and that maximize Log REML are shown as vertical dotted lines. Each row corresponds to one choice of principal components, hence one variance proportion. (b) is the plot of negative log restricted maximum likelihood. Vertical dotted line indicates the value of penalty that minimizes the negative REML curve. Each row corresponds to one choice of principal components, hence one variance proportion.}
    \label{fig:FPCR}  
\end{figure}

\section*{Acknowledgment}
The authors gratefully acknowledge some helpful discussions of this work with Dr. Tianyi Pan. 

\section*{Supplementary Materials}
Code to reproduce the simulation results and figures in this paper is available at \url{https://github.com/HenrySyangC/Fast-NCV-paper-code.git}.

\bibliographystyle{apalike}
\bibliography{ncv_cite}

\clearpage
\appendix
\section{Justification of leave-several-out formula}
\label{appendix-proof-leave-several-out}
We prove the residual identity that we took advantage of, in particular equation (\ref{eqn:residualidentity}). We first state it as a lemma and then give a complete proof.
\begin{lemma}
    \label{lemma:leave-several-out-formula}
We consider a penalized regression setting. Let $H(\lambda) \in \R^{n \times n}$ be the hat matrix, $\II \subseteq \{1, \cdots, n\}$ and $I \in \R^{|\II| \times |\II|}$ be the identity matrix. We follow the same subsetting matrix and vector notations given before. Let $\hat{y}$ be the predicted $y$ from the full model fit. In particular, $\hat{y} = H(\lambda)y$. Let $\hat{y}^{-\II}$ be the predicted $y$ by fitting the penalized regression on data that are not in $\II$, meaning $(X_{-\II\cdot}, y_{-\II})$. Then \begin{equation}
    \left(I - H(\lambda)_{\II\II}\right)^{-1}\left(y_\II - \hat{y}_\II\right) = y_\II - \hat{y}_\II^{-\II}. \label{eqn:leave-several-out}
\end{equation}
\end{lemma}
\begin{proof}
    We adapt a proof from \citet{wood2017}, \citet{Hastie01102020} and \citet{Golub01051979}. When fitting a penalized regression model on $(X_{-\II\cdot}, y_{-\II})$, we want to solve the optimization problem as below \begin{equation}\begin{aligned}
        \min_{\beta}\quad \left(y_{-\II} - X_{-\II\cdot}\beta\right)^T\left(y_{-\II} - X_{-\II\cdot}\beta\right)  + \lambda \beta^TS\beta \label{eqn:leave-several-out-opt}
    \end{aligned}\end{equation}
    where $\hat{\beta}^{-\II}_{\lambda} = (X_{-\II\cdot}^T X_{-\II\cdot} + \lambda S)^{-1}X^T_{-\II\cdot}y_{-\II}$ is the fitted coefficient of penalized regression on $(X_{-\II\cdot}, y_{-\II})$. Now let $\hat{y}_\II^{-\II} = X_{\II\cdot}\hat{\beta}^{-\II}_\lambda$. If we change (\ref{eqn:leave-several-out-opt}) to \begin{equation}
        \begin{aligned}
                    \min_{\beta}\quad \left(y_{-\II} - X_{-\II\cdot}\beta\right)^T\left(y_{-\II} - X_{-\II\cdot}\beta\right) + \left(\hat{y}_\II^{-\II} - X_{\II\cdot}\hat{\beta}^{-\II}_\lambda\right)^T \left(\hat{y}_\II^{-\II} - X_{\II\cdot}\hat{\beta}^{-\II}_\lambda\right)+ \lambda \beta^TS\beta \label{eqn:leave-several-out-optadj}, 
        \end{aligned}
    \end{equation}
    the solution would still be the same. Observe that (\ref{eqn:leave-several-out-optadj}) is equivalent to fitting a penalized regression on $X$ and $(\hat{y}_\II^{-\II}, y_{-\II})$. Since we fit it on $X$, the hat matrix is $H(\lambda) = X(X^TX + \lambda S)^{-1}X^T$. Let $\hat{y} = H(\lambda)y$, the predicted $y$ for fitting the model on $(X, y)$. Observe that by (\ref{eqn:leave-several-out-optadj}) and the definition of $\hat{y}_\II^{-\II}$, we have $$H(\lambda)_{\II\cdot}\begin{pmatrix}
        \hat{y}_\II^{-\II} \\ y_{-\II}
    \end{pmatrix} = H(\lambda)_{\II\II}\hat{y}^{-\II}_\II + H(\lambda)_{\II(-\II)}y_{-\II} = \hat{y}^{-\II}_\II. $$ It follows that \begin{align*}
        &H(\lambda)_{\II\II}\hat{y}^{-\II}_\II + H(\lambda)_{\II\II} y_\II + H(\lambda)_{\II(-\II)}y_{-\II} - H(\lambda)_{\II\II}y_{\II} = \hat{y}^{-\II}_\II \\
        \Rightarrow \quad & \hat{y}_\II + H(\lambda)_{\II\II}\hat{y}_\II^{-\II} - H(\lambda)_{\II\II}y_{\II} = \hat{y}^{-\II}_\II \\
        \Rightarrow \quad &  -\hat{y}^{-\II}_\II - H(\lambda)_{\II\II}\left(y_\II - \hat{y}^{-\II}_\II\right) = -\hat{y}_\II \\
        \Rightarrow \quad & y_\II -\hat{y}^{-\II}_\II - H(\lambda)_{\II\II}\left(y_\II - \hat{y}^{-\II}_\II\right) = y_\II -\hat{y}_\II \\
        \Rightarrow \quad & \left(I - H(\lambda)_{\II\II}\right)\left(y_\II - \hat{y}^{-\II}_\II\right) = y_{\II} - \hat{y}_\II \\
        \Rightarrow \quad & y_\II - \hat{y}^{-\II}_\II = \left(I - H(\lambda)_{\II\II}\right)^{-1}\left(y_{\II} - \hat{y}_\II\right)
    \end{align*}
    as desired. 
\end{proof}

\section{Justification of Proposition \ref{prop:finitesample}}
\label{appendix-proof-prop-2}
\begin{proof}[Proof of Proposition \ref{prop:finitesample}]
    We first outline our strategy. Given $n, p, K$, we want to count the FLOPs for implementing our proposed formula (Proposition \ref{ncvformula}) and the FLOPs for the old method, which is refitting the model. Both will give us polynomial of $K$ with coefficients depending on $n$ and $p$. Equate them and rearrange to get $f(K)$. We summarize our FLOP count of computing $\res{jk}$ for these two methods in Table \ref{tab:flop-count-formula} and Table \ref{tab:flop-count-refit}. 
\begin{table}[ht]
\centering
\caption{FLOP Count by Step--Proposed Formula}
\label{tab:flop-count-formula}
\begin{tabular}{>{\raggedright\arraybackslash}p{0.5\textwidth} c p{0.13\textwidth}}
\toprule
\textbf{Computation Step} & \textbf{FLOP Count}  \\
\midrule
 $\left(X^TX + \lambda S\right)^{-1}X_{\II_{jk}\cdot}^T$ & $2p^2(2n/K)$ \\
$X_{\II_{jk}\cdot}\left(X^TX + \lambda S\right)^{-1}X_{\II_{jk}\cdot}^T$ & $2p(2n/K)^2$ \\
$I - X_{\II_{jk}\cdot}\left(X^TX + \lambda S\right)^{-1}X_{\II_{kj}\cdot}^T$ & $(2n/K)^2$ \\
$\left(I - X_{\II_{jk}\cdot}\left(X^TX + \lambda S\right)^{-1}X_{\II_{jk}\cdot}^T\right)^{-1}$ & $\frac{1}{3}(2n/K)^3 + 2(2n/K)^2$  \\
$\left(I - X_{\II_{jk}\cdot}\left(X^TX + \lambda S\right)^{-1}X_{\II_{jk}\cdot}^T\right)^{-1}(y_{\II_{kj}} - \hat{y}_{\II_{jk}})$ & $2(2n/K)^2$ \\
\midrule
\multicolumn{1}{c}{\textbf{Total}} & $\begin{aligned}
& 2p^2(2n/K) + 2p(2n/K)^2 \\
& + \frac{1}{3}(2n/K)^3 + 5(2n/K)^2
\end{aligned}$ \\
\bottomrule
\end{tabular}
\end{table}

\begin{table}[ht]
\centering
\caption{FLOP Count by Step--Refit the Model}
\label{tab:flop-count-refit}
\begin{tabular}{>{\raggedright\arraybackslash}p{0.48\textwidth} c p{0.145\textwidth}}
\toprule
\textbf{Computation Step} & \textbf{FLOP Count}  \\
\midrule
$X_{-\II_{jk}\cdot}^TX_{-\II_{jk}\cdot}$ & $2p^2(n - 2n/K)$ \\
$X_{-\II_{jk}\cdot}^TX_{-\II_{jk}\cdot} + \lambda S$ & $p^2$ \\
$\left(X_{-\II_{jk}\cdot}^TX_{-\II_{jk}\cdot} + \lambda S\right)^{-1}$ & $\frac{1}{3}p^3 + 2p^2$  \\
$X^T_{-\II_{jk}\cdot}y_{-\II_{jk}}$ & $2p(n - 2n/K)$ \\
$\left(X_{-\II_{jk}\cdot}^TX_{-\II_{jk}\cdot} + \lambda S\right)^{-1}X^T_{-\II_{jk}\cdot}y_{-\II_{jk}}$ & $2p^2$  \\
$X_{\II_{jk}\cdot}\hat{\beta}$ & $2p(2n/K)$  \\
$y_{\II_{jk}} - X_{\II_{jk}\cdot}\hat{\beta}$ & $2n/K$  \\
\midrule
\multicolumn{1}{c}{\textbf{Total}} & $\begin{aligned}
&2p^2(n - 2n/K) + \frac{1}{3}p^3 + 5p^2 \\ 
&+ 2pn + 2n/K
\end{aligned}$ \\
\bottomrule
\end{tabular}
\end{table}

Notice that both Table \ref{tab:flop-count-formula} and Table \ref{tab:flop-count-refit} uses a pair of folds $\II_{jk}$ setting. When it is for one fold $\II_j$, we replace all of the $2n/K$ in Table \ref{tab:flop-count-formula} and Table \ref{tab:flop-count-refit} with $n / K$.  From section \ref{subsec:nestedcrossvalidation} and section \ref{sec:fast-ncv}, we need to compute $K(K-1)/2$  double folds residuals $\res{jk}$ and compute $K$ one fold residual $\res{k}$. Thus, the total FLOP count for proposed formula is \begin{align*}
    f_1(K) =& \left[2p^2 \frac{2n}{K} + 2p\left(\frac{2n}{K}\right)^2 + \frac{1}{3}\left(\frac{2n}{K}\right)^3 + 5\left(\frac{2n}{K}\right)^2 \right]\cdot \frac{K(K-1)}{2} + \\& \left[2p^2\frac{n}{K} + 2p\left( \frac{n}{K}\right)^2 + \frac{1}{3}\left(\frac{n}{K}\right)^3 + 5\left(\frac{n}{K}\right)^2 \right]\cdot K. 
\end{align*}
The total FLOP count for reftting the model is \begin{align*}
    f_2(K) =& \left[2p^2 n - 2p^2 \frac{2n}{K} + \frac{1}{3}p^3 + 5p^2 + 2pn + \frac{2n}{K} \right] \cdot \frac{K(K-1)}{2} + \\ &\left[2p^2n - 2p^2 \frac{n}{K} + \frac{1}{3}p^3 + 5p^2 + 2pn + \frac{n}{K} \right]\cdot K. 
\end{align*}
Equate $f_1(K) = f_2(K)$ and rearrange as $f_2(K) - f_1(K)$. We multiply by $6K^2$ on both sides of $f_1(K) = f_2(K)$ to get rid of all the fractions. Then some algebra in between is omitted. We reach \begin{equation*}
\begin{aligned}
    f(K) =& \alpha^{(0)}_{n,p} + \alpha^{(1)}_{n,p}K + \alpha^{(2)}_{n,p}K^2 + \alpha^{(3)}_{n,p}K^3 + \alpha^{(4)}_{n,p}K^4, \\
    \alpha^{(0)}_{n,p} &= 6n^3, \\
\alpha^{(1)}_{n,p} &= -8n^3 + 12pn^2 + 30n^2, \\
\alpha^{(2)}_{n,p} &= -24pn^2 - 60n^2, \\
\alpha^{(3)}_{n,p} &= -18p^2n + p^3 + 15p^2 + 6n + 6pn, \\
\alpha^{(4)}_{n,p} &= 6p^2n + p^3 + 15p^2 + 6pn.
\end{aligned}
\end{equation*}
Suppose we have a root $K^* \in [2, n]$ for $f(K)$. When $K > K^*$, $f_2(K) - f_1(K) > 0$, thus $f_1(K) < f_2(K)$, meaning our formula is faster and vice versa.
\end{proof}

\section{Restricted maximum likelihood for functional principle components regression}
\label{appendix-REML}
Restricted maximum likelihood starts by looking at functional principle components regression model in (13) as \begin{align*}
    y|\zeta &\sim N(\alpha \mathbf{1} + XBV_p \zeta, \sigma^2I), \\
    \zeta &\sim N(\mathbf{0}, (\sigma^2/\lambda)W^{-1})
\end{align*}
where we define \begin{align*}
    W := V_p^T SV_p. 
\end{align*}
Note that $\operatorname{rank}(S) = 38$ as we are taking $40$ spline basis, and we are penalizing second-order roughness. We choose $p < 38$, so that $W$ is invertible. Restricted maximum likelihood selects $\lambda$ by maximizing the restricted log-likelihood \begin{align}
    l_R(\alpha, \lambda, \sigma | y) = -\frac{1}{2}\{\log|\sigma^2 V_\lambda| + (y - \alpha\mathbf{1})^T(\sigma^2 V_\lambda)^{-1}(y - \alpha \mathbf{1}) + \log|\sigma^{-2}\mathbf{1}^T V_\lambda^{-1}\mathbf{1}|\} \label{appendix-fpcr-eq4}
\end{align}
where we define\begin{align*}
    V_\lambda := I + \lambda^{-1} (XBV_p)^TW^{-1}(XBV_p). 
\end{align*}
Note that $V_\lambda$ is invertible since $W$ is invertible and $V_p$ comes from the first $p$ eigenvectors of $XB$. Given $\alpha, \lambda$, the value of $\sigma^2$ maximizing (\ref{appendix-fpcr-eq4}) is \begin{align*}
    \widehat{\sigma}^2_{\alpha, \lambda} =\frac{1}{n-1} (y - \alpha \mathbf{1})^T V_{\lambda}^{-1}(y - \alpha \mathbf{1}).
\end{align*} It follows that we maximize the profile restricted log-likelihood \begin{align}
    l_R(\alpha, \lambda | y) = -\frac{1}{2}[\log|V_\lambda| + \log|\mathbf{1}^TV_\lambda^{-1}\mathbf{1}| + (n - 1)\log\{(y - \alpha \mathbf{1})^TV_\lambda^{-1}(y - \alpha \mathbf{1})\}]. \label{appendix-fpcr-eq5}
\end{align}
Note that since we centered $y$ and $X$, we know $\alpha = 0$ and (\ref{appendix-fpcr-eq5}) reduces to \begin{align*}
    l_R(\lambda | y) = -\frac{1}{2}[\log|V_\lambda| + \log|\mathbf{1}^TV_\lambda^{-1}\mathbf{1}| + (n - 1)\log\{y^TV_\lambda^{-1}y\}]. 
\end{align*}

\end{document}